\RequirePackage{fix-cm}
\documentclass[smallextended]{svjour3}       % onecolumn (second format)
\usepackage{graphicx}
\journalname{
}
\usepackage{amsmath}	% align package
\usepackage{graphicx} %to load the graphicx package
\usepackage{amsfonts}
\usepackage{algorithm}
\usepackage{algorithmic}
\usepackage{bm} % for both text and math in Latin Modern

\usepackage[rightcaption]{sidecap}
\usepackage{float}
\usepackage{url}
\usepackage{fancyhdr}
\usepackage{cite}
\usepackage{enumerate}
\usepackage{bm} % for both text and math in Latin Modern
\newcommand{\nint}{\mathbb{N}_0}

\newcommand{\cln}{\text{:\;}}
\newcommand{\Ga}{\mathcal{G}_{m,\sigma}}
\newcommand{\Gl}{\mathcal{G}_{m < \sigma}}
\newcommand{\Gg}{\mathcal{G}_{m \geq \sigma}}
\usepackage[pdfpagemode={UseOutlines},bookmarks=true,bookmarksopen=true,
   bookmarksopenlevel=0,bookmarksnumbered=true,hypertexnames=false,
   colorlinks=true,linkcolor={blue},citecolor={red},urlcolor={magenta},
   pdfstartview={FitV},unicode,breaklinks=true, pdftex]{hyperref}
\usepackage{graphicx}\graphicspath{{../../../thesis/fig/pdf/}{../../../thesis/fig/eps/}{../../../thesis/fig/1_3/}{../../../thesis/fig/3_1/}{../../../thesis/fig/4_1/}{../../../thesis/fig/5_1/}{../../../thesis/fig/5_2/}{../../../thesis/fig/6_1/}{../../../thesis/fig/capsm/}{../../../thesis/fig/literature/}{../../../thesis/fig/mwbm/}}

\begin{document}

\title{A Tight Lower Bound for the Weights of Maximum Weight Matching in Bipartite Graphs
}

\titlerunning{A Tight LB for the Weights of Maximum Weight Matching in Bipartite Graphs}        % if too long for running head

\author{Shibsankar Das}


\institute{Shibsankar Das\at 
			Department of Mathematics,
			Institute of Science, Banaras Hindu University,
			Varanasi - 221\ 005, Uttar Pradesh, India.\\
              E-mail: shib.iitm@gmail.com, shibsankar@bhu.ac.in
}

\date{%Received: \today / Accepted: date
}
\maketitle
%================================================================================================
\begin{abstract}
%We provide a tight lower bound  for % ON is also OK
%the weights of maximum weight  matching of bipartite graphs having  fixed weight and vertex size.
%

Let $\Ga$ %$\mathcal{G}$
%$\mathcal{G}_{m \geq \sigma} / \mathcal{G} (\mathcal{G}', \mathcal{G}'')$ 
be the collection of all weighted bipartite graphs each having   $\sigma$ and $m$, as the size of a vertex partition and the total weight, respectively. % of the graphs in $\Ga$.
We give a tight lower bound  $\lceil  \frac{m-\sigma}{\sigma} \rceil+1$ for the set $\{\textit{Wt}(\textit{mwm}(G))~|~G \in \Ga\}$ which denotes the collection of weights of maximum weight bipartite matchings of all graphs in  $\Ga$.
%
%
%each of whose weight is  % fixed to 
%$m$ and $\sigma$ is the size of each vertex partition of any graph in $\mathcal{G}$.
%
%$\textit{G}=(\varSigma_P \cup \varSigma_T,E,\textit{Wt})$ where $\sigma = |\varSigma_P|=|\varSigma_T|,~ m=\textit{Wt}(G)$ and $ m \geq \sigma$ where $\textit{Wt}$ is a weight function ??? and $\textit{Wt}(G)$ is the weight of $G$. 
%
%Further, let 
%the set $\{\textit{Wt}(\textit{mwm}(G))~|~G \in \Ga\}$ %_{m \geq \sigma} 
% denotes the collection of weights of maximum weight bipartite matchings of all the graphs in  $\Ga$. %_{m \geq \sigma}$.  
%%minimum weight of maximum weight bipartite matching of $G$ over all $G \in \mathcal{G}$.
%%
%%
%We provide a tight lower bound  for % ON is also OK
%the weights of maximum weight  matching of bipartite graphs having  fixed weight and vertex size as $\lceil  \frac{m-\sigma}{\sigma} \rceil+1$ which is a minimum of $\{\textit{Wt}(\textit{mwm}(G))~|~G \in \Ga\}$. %_{m \geq \sigma} 
%% and in fact this bound is tight.
\end{abstract}
%================================================================================================
%================================================================================================
\section{Introduction}
%\todo{$m  < \sigma$ case, the value is 1.}
%\subsection{Basics of Maximum Matching in a Bipartite Graph}
\label{Basics_MWBM}
%As pointed out earlier in Section~\ref{lit:Sec:APSM}, the problem of APSM for a pair of equal length strings under Hamming distance is computationally equivalent to Maximum Weight Bipartite Matching (MWBM) problem in graph theory, due to reduction between the problems given by Hazay et al.~\cite{hazay07}. See Section~\ref{Reduction_APSM_MWBM} for reduction idea.

%\subsection{Preliminaries and Definitions}
We use the notations $\mathbb{N}$ and $\nint$ to denote the sets of positive integers and non-negative integers, respectively.
Let $G=(V = V_1 \cup V_2, E, \textit{Wt})$
be an undirected, weighted bipartite graph 
%without isolated nodes 
where $V_1$ and $V_2$ are two non-empty partitions of the vertex set $V$ of $G$, and $E$ is the edge set of $G$ with
% non-negative 
positive integer weights on the edges which are given by the weight function 
%$\textit{Wt}\cln E \rightarrow \nint$.
$\textit{Wt}\cln E \rightarrow \mathbb{N}$.
%
%Let $G=(V, E, \textit{Wt})$ be an undirected and weighted graph with $V$ and $E$ as the set of vertices and edges, respectively, and has non-negative integer weights on the edges which are given by the weight function $\textit{Wt}\cln E \rightarrow \nint$, where $\nint$ is the set of non-negative integers.
Let $W$ denotes the total weight of $G$ and 
%The weight of the graph $G$ 
is defined by
%$W=\textit{Wt}(E) = \sum_{e \in E} \textit{Wt}(e)$.
$W=\textit{Wt}(G)=  \sum_{e \in E} \textit{Wt}(e)$. 
%We also assume that the graph does not have any isolated vertex. 
For uniformity, we treat an unweighted
graph as a weighted graph having a unit weight for all edges.

\subsection{Basics of Maximum Weight Bipartite Matching}
%Let $G=(V,E,\textit{Wt})$ be a undirected graph. 
We use the notation $\{u,v\}$ for an edge $e \in E$ between $u \in V_1$ and $v \in V_2$, and its weight is denoted by $\textit{Wt}(e)=\textit{Wt}(u,v)$.
%If $e=\{u,v\}\in E$, 
We also say that $e=\{u,v\}$ is \textit{incident} on vertices $u$ and $v$,
%; also we say that 
and $u$ and $v$ are each \textit{incident} with $e$.
Two vertices $u,v \in V$ of $G$ are \textit{adjacent} if there exists an edge $e=\{u,v\}\in E$ of $G$ to which they are both incident. 
Two edges $e_1,e_2 \in E$ of $G$  are  \textit{adjacent} if there exists a vertex $v\in V$ to which they are both incident~\cite{cormen01}.
%\cite{cormen01,douglas00} ????Check

A subset $M \subseteq E$ of edges is a \emph{matching} if no two edges of $M$
share a common vertex. A vertex $v \in V$ is said to be \emph{covered} or
\emph{matched} by the matching $M$ if it is incident with an edge of
$M$; otherwise $v$ is \emph{unmatched}~\cite{bondy82,bondy08}.
A matching $M$ of $G$ is called a \textit{maximum} (\textit{cardinality}) \textit{matching} if there
does not exist any other matching of $G$ with greater cardinality. We denote such a
matching by $\textit{mm}(G)$. The weight of a matching $M$ is defined as
$\textit{Wt}(M) = \sum_{e \in M} \textit{Wt}(e)$. A matching $M$ of $G$ is a \emph{maximum weight
matching}, denoted as $mwm(G)$, if $\textit{Wt}(M) \geq \textit{Wt}(M')$ for every other matching
$M'$ of the graph $G$.
Observe that, if $G$ is an unweighted graph then a $\textit{mwm}(G)$ is a $\textit{mm}(G)$, which we write as $\textit{mwm}(G)=\textit{mm}(G)$ in short and its
weight is given by $\textit{Wt}(\textit{mwm}(G))$ $=|\textit{mm}(G)|$. Similarly, if $G$ is an
undirected and weighted graph with $\textit{Wt}(e) = c$ for all edges $e$ in $G$
and $c$ is a constant then also we have $\textit{mwm}(G)=\textit{mm}(G)$ with weight of the
matching as $\textit{Wt}(\textit{mwm}(G))=c*|\textit{mm}(G)|$.

%Let $G=(V = V_1 \cup V_2, E, \textit{Wt})$ be an undirected, weighted bipartite graph and without isolated nodes and having $V_1$ and $V_2$ as partition of vertex set $V$.

Maximum Weight Bipartite Matching (MWBM) problem is a well-studied problem in combinatorial optimization and algorithmics, and has a wide range of applications (see 
%standard
 textbooks~\cite{schrijver03,korte07}). 
Several exact, approximate and randomized algorithms have also been proposed for computing maximum weight bipartite matching~\cite{gabow85,fredman87,gabow89,kao02,das14,duan10,sankowski09}. %\cite{}

%\subparagraph{\it Our Contribution.}
%Our contribution 
%{\it Our Contribution.}
\subsection{Our Contribution}
In this paper, we give a tight lower bound for the weights of MWBM 
%Maximum Weight Bipartite Matching (MWBM) of
in bipartite graphs
having fixed weight and vertex size. 
Let $\Ga$
%$\mathcal{G}_{m \geq \sigma} / \mathcal{G} (\mathcal{G}', \mathcal{G}'')$ 
be the collection of all weighted bipartite graphs, each of whose weight is $m$ and $\sigma$ is the size of each   partition of the vertex set, where $m$ and $\sigma$ are positive integers.
% of any graph in $\mathcal{G}$.
%Let $\mathcal{G}$ be the collection of all weighted bipartite graphs and 
The set of weights of MWBM of the graphs in  $\Ga$ is denoted by $\{\textit{Wt}(\textit{mwm}(G))~|~G \in \Ga\}$. %_{m \geq \sigma} 
 %_{m \geq \sigma}$.  
%minimum weight of maximum weight bipartite matching of $G$ over all $G \in \mathcal{G}$.
%
%
We prove that $\lceil  \frac{m-\sigma}{\sigma} \rceil+1$ is a lower bound of $\{\textit{Wt}(\textit{mwm}(G))~|~G \in \Ga\}$ 
and this bound is tight. %_{m \geq \sigma} 
% and in fact this bound is tight.
%{\bf This problem has a direct application in the theory of sting matching}.
%Motivation of this problem is
%The MWBM problem is computationally equivalent to  the problem of Approximate Parameterized String Matching (APSM) for a pair of equal length strings under Hamming distance~\cite{hazay07}.
%
%Finding a tight lower bound for the weights of MWBM has direct application in combinatorial properties of the error classes in APSM problem. 
%
%Let $\mathcal{G}_{m \geq \sigma}$ be the collection of all weighted bipartite graphs $\textit{G}=(\varSigma_P \cup \varSigma_T,E,\textit{Wt})$ where  $\sigma = |\varSigma_P|=|\varSigma_T|, m=\textit{Wt}(G)$ and $ m \geq \sigma$ where $\textit{Wt}$ is a weight function ??? and $\textit{Wt}(G)$ is the weight of $G$. 
%
%Further, let the set $\{\textit{Wt}(\textit{mwm}(G))~|~G \in \mathcal{G}_{m \geq \sigma} \}$ denotes the collection of weights of maximum weight bipartite matchings of all the graphs in  $\mathcal{G}_{m \geq \sigma}$. 
%%minimum weight of maximum weight bipartite matching of $G$ over all $G \in \mathcal{G}$.
%%
%%
%Then $\lceil  \frac{m-\sigma}{\sigma} \rceil+1$ is a lower bound of $\{\textit{Wt}(\textit{mwm}(G))~|~G \in \mathcal{G}_{m \geq \sigma} \}$ and in fact this bound is tight.

This result can be applied to develop 
%telecommunication network (with minimum cost, minimum time, minimum traffic load analysis, critical path routing), monitoring computer network, computer vision, pattern recognition, machine learning~\cite{bunke00}, 
optical packet switches to transform data center scalability~\cite{benjamin17}. It can also pertain to the area of stringology~\cite{hazay07}.
An equivalent theorem (Theorem 3.14 in \cite{mythesis}) of this outcome is applied for enumerating error classes in approximate     parameterized string matching.

%also in compiler design with cloud architecture.  

%----------------------------------------------------------------

%\subparagraph{\it Roadmap.}
%{\it Roadmap.}
\subsection{Roadmap}
The rest of the paper is organized as follows. 
In Section~\ref{capsm:Sec:Tight lower bound_MWBM},
%\ref{Sec:tlb_MWBM} 
we partition the class of graphs in $\Ga$ into two subclasses and 
%address the main problem of this paper. We
provide a tight lower bound for the weights of MWBM of graphs in $\Ga$.
A summary  is given in Section~\ref{Sec:tlb_conclusion}.
%Preliminaries of maximum matching in a bipartite graph is given in Section~\ref{Basics_MWBM}. Section~\ref{Reduction_APSM_MWBM} describes the reductions between the APSM problem under Hamming distance and the MWBM problem. We provide a tight lower bound on the weights of MWBM of graphs in Section~\ref{capsm:Sec:Tight lower bound_MWBM} which is used in the proof of counting the number of error classes in APSM problem. In Section~\ref{capsm:Sec:ECsAPSM} we introduce the term Error Class (EC) in APSM. Several combinatorial properties of these classes, such as distribution of empty and non-empty ECs, exact count of  non-empty ECs and divisibility property of ECs in APSM are presented in the Sections~\ref{capsm:Sec:Attribute ECsAPSM}--\ref{Sec:capsm:Miscellaneous}. An experimental result is given in Section~\ref{Sec:capsm:Enumeration} which enumerates the number of pairs in an EC. Finally, a summary of this chapter is given in Section~\ref{capsm:conclusion}. 
%================================================================================================
\section{A Tight Lower Bound for the Weights of Maximum Weight Bipartite Matching in $\Ga$}
%\section{A Tight Lower Bound on the Weights of MWBM of Graphs}% in Graph Theory}
\label{capsm:Sec:Tight lower bound_MWBM}
Let $\Ga$ denotes the collection of all weighted bipartite graphs, each of whose weight is fixed to $m$ and $\sigma$ is the size of each of the two vertex partitions of any graph in $\Ga$.
Let us consider the pair of non-empty partitions of the vertex sets of all the bipartite graphs in $\Ga$ be $\varSigma_P$ and $\varSigma_T$. Therefore, %where 
$|\varSigma_P|=|\varSigma_T|=\sigma$.
%, and the sets $P$ and $T$ denote the set of all $m$ length patterns and texts over $\varSigma_P$ and $\varSigma_T$, respectively,  that is $P=\varSigma_P^m$ and $T=\varSigma_T^m$.
%More formally, $P=\varSigma_P^m$ and $T=\varSigma_T^m$.
%
%
We partition $\Ga$ as $\Ga = \mathcal{G}_{m \geq \sigma} \cup \mathcal{G}_{m < \sigma}$, 
%$\Gall = \Gl \cup \Gg$
where
%
%Further, let $\mathcal{G}_{m \geq \sigma}$ be the collection of all weighted bipartite graphs 
%$\textit{G}=(\varSigma_P \cup \varSigma_T,E,\textit{Wt})$ where 
%$\sigma = |\varSigma_P|=|\varSigma_T|, m=\textit{Wt}(G)$ and $ m \geq \sigma$. For notational convenience we use $\mathcal{G}$ instead of $\mathcal{G}_{m \geq \sigma}$. Therefore,
%
$$ \Gg \equiv \{\textit{G}=(\varSigma_P \cup \varSigma_T,E,\textit{Wt})~|~ \sigma = |\varSigma_P|=|\varSigma_T|,\ m=\textit{Wt}(G),\ m \geq \sigma \}$$
and
$$ \Gl \equiv \{\textit{G}=(\varSigma_P \cup \varSigma_T,E,\textit{Wt})~|~ \sigma = |\varSigma_P|=|\varSigma_T|,\ m=\textit{Wt}(G),\ m < \sigma \}.$$
Now we prove that the value of $\min_{G \in \Ga}\{ \textit{Wt}(\textit{mwm}(G)) \}$, which denotes minimum weight among the maximum weight bipartite matchings of all the graphs in  $\Ga$, is $ \lceil  \frac{m-\sigma}{\sigma} \rceil+1$. 
%minimum weight of maximum weight bipartite matching of $G$ over all $G \in \mathcal{G}$.
%Therefore, $\min_{G \in \mathcal{G}}\{ \textit{Wt}(\textit{mwm}(G))\}$ is a lower bound of $\{\textit{Wt}(\textit{mwm}(G))~|~G \in \mathcal{G} \}$. % for any bipartite graph $G \in \mathcal{G}$.
% and in fact this bound is tight.
Let us first prove it for $\Gg \subseteq \Ga$.

Since $m  \geq \sigma$, we can always write $m$ as 
% a linear combination of $\sigma$ and~$1$. %Let $m=
$q \sigma + r$ for some  $q,\,r \in \nint$  where $0 < r \leq \sigma$.
%, we have $\mathcal{G}=\{\textit{G}(\varSigma_P \cup \varSigma_T,E,\textit{Wt})~|~ \sigma = |\varSigma_P|=|\varSigma_T|, \textit{Wt}(E)=m=q \sigma + r \text{ and } m \geq \sigma \}$.
First, we show the existence of bipartite graph $G \in \Gg$ such that $\textit{Wt}(\textit{mwm}(G))= q+1$.
% in Theorem~\ref{capsm:Th:min_MWBMs_Existence}. 
We then prove  in Theorem~\ref{capsm:Th:min_MWBMs} that $q+1$ is the tight lower bound of the set $\{ \textit{Wt}(\textit{mwm}(G))~|~G\in \Gg\}$.
\begin{theorem}
%[Existence of Bipartite Graph with Weight of the MWBM Equal to $q+1$]
\label{capsm:Th:min_MWBMs_Existence}
Let $\Gg=\{G=(\varSigma_P \cup \varSigma_T,E,\textit{Wt})~|~ \sigma = |\varSigma_P|=|\varSigma_T|,\, m=\textit{Wt}(G) \text{ and } m \geq \sigma \}$.
If $m=q\sigma + r$ for some non-negative integers  $q$ and  $r$ where $0 < r \leq \sigma$,
then there exists %at least 
a bipartite graph $G \in \Gg$ such that $\textit{Wt}(\textit{mwm}(G))= q+1$.
\end{theorem}
\begin{proof}
\begin{figure*}[!h]%[!hpbt]
\centering
\includegraphics[angle=0,width=.9275\textwidth, trim=0 -5 0 5]{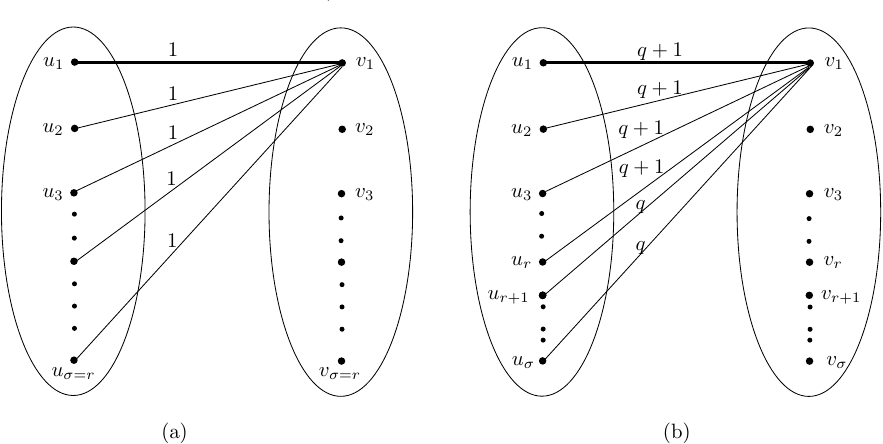}
\caption[Existence of bipartite graph $G \in \mathcal{G}=\{\textit{G}=(\varSigma_P \cup \varSigma_T,E,\textit{Wt})~|~ \sigma = |\varSigma_P|=|\varSigma_T|, m=\textit{Wt}(G), m \geq \sigma \}$ such that $\textit{Wt}(\textit{mwm}(G))= q+1$
where $ m=q\sigma + r$ for some $q,r \in \nint$ and $0 < r \leq \sigma$.]
{Given $ m=q\sigma + r$ for some $q,r \in \nint$ where $0 < r \leq \sigma$ and $\Gg =\{G=(\varSigma_P \cup \varSigma_T,E,\textit{Wt})~|~ \sigma = |\varSigma_P|=|\varSigma_T|,\, m=\textit{Wt}(G),\, m \geq \sigma \}$ such that
 $\min_{G \in \Gg}\{ \textit{Wt}(\textit{mwm}(G)) \} = q+1$. {\bf(a)} An example of bipartite graph for the case $q=0$. {\bf(b)} 
 An example of bipartite graph for the case
 $q \geq 1$. 
 In both the graphs the thick edge represents maximum weight matching edge.}
\label{capsm:Fig2}
\end{figure*}
%
%The proof is existential{\red ???}. 
%In Figure~\ref{capsm:Fig2}$(b)$, we produce a bipartite graph $G \in \mathcal{G}$ 
%%satisfying the required properties 
%where $\textit{Wt}(\textit{mwm}(G))= q+1$.
For the case $q=0$, we have $m=q\sigma+r=r=\sigma$ as $0 < r \leq \sigma$ and $m \geq \sigma$. 
Figure~\ref{capsm:Fig2}(a) shows a bipartite graph $G'=(\varSigma_P \cup \varSigma_T,E',\textit{Wt}) \in \Gg$ for this case. The weight of the graph is $\textit{Wt}(G')=\sigma$. 
In this graph $G'$, $\textit{Wt}(\textit{mwm}(G'))= 1=q+1$.

For $q \geq 1$, the total weight of any bipartite graph in $\Gg$ is $ m=q\sigma + r$. We produce such a bipartite graph $G'' \in \Gg$  shown in 
Figure~\ref{capsm:Fig2}(b)
% shows such an example and in this case 
with $\textit{Wt}(\textit{mwm}(G''))= q+1$.
\end{proof}

%
%
%\begin{observation}
%\label{capsm:Obs1}
Observe that in a weighted graph $G$, any edge $e$ of weight $c \in \mathbb N$ can be thought of as $c$ number of overlapping unit weight edges.
Similarly, increasing the weight of a bipartite graph $G$ by adding a weight $c \in \mathbb{N}$ is equivalent to adding $c$  unit weight edges in $G$. 
%\end{observation}
Without loss of generality, we assume these as  a convention while incrementing weight in a weighted graph. 
%It is mentioned below.
%\begin{observation}
%\label{capsm:Obs2}
%\end{observation}

\begin{theorem}[Tight Lower Bound for the Weights of MWBM of the Graphs in $\Gg$]
\label{capsm:Th:min_MWBMs}
Let $\Gg=\{\textit{G}=(\varSigma_P \cup \varSigma_T,E,\textit{Wt})~|~ \sigma = |\varSigma_P|=|\varSigma_T|,\, m=\textit{Wt}(G) \text{ and } m \geq \sigma \}$. Then
$$ \min_{G \in \Gg}\{ \textit{Wt}(\textit{mwm}(G)) \} = q+1 $$ 
where $m=q\sigma + r$ for some non-negative integers  $q$ and  $r$, and $0 < r \leq \sigma.$
\end{theorem}
\begin{proof}
For $\sigma=1$, 
%that is, the number of vertex in a partition is one, then 
the statement is trivially true. 
% When $\sigma=1$, therefore $r=1$, as $0<r \leq \sigma$. So, Wt(E)=m=q+1
%
%
%
So we consider $\sigma \geq 2$ and prove the statement $\min_{G \in \Gg}\{ \textit{Wt}(\textit{mwm}(G)) \} = q+1$ by
% using the principle of mathematical 
induction on $q \in \nint$.
%As per the construction of the bipartite graph,  alphabets $\varSigma_P$ and $\varSigma_T$ of the strings are the two disjoint vertex sets of the graph. 
Let $\varSigma_P = \{u_1,u_2, \ldots, u_\sigma\}$ and $\varSigma_T = \{v_1,v_2, \ldots, v_\sigma\}$ be the disjoint vertex sets of the graphs in $\Gg$. 
For simplicity, we denote $\mathcal{G}_{q+1} =\Gg$ when $m=q\sigma+r$ for some $q,r \in \nint$ where $0 < r \leq \sigma$, that is,  $q=\lceil  \frac{m-\sigma}{\sigma}  \rceil$  where $q$ is represented as a function of $m$ and $\sigma$ only.
% as well as the alphabet sets for APSM problem.
%

\begin{description}
\item[Base Step:] 
Let $q=0$. Then $m=r = \sigma$ because $0 < r \leq \sigma$ and $m \geq \sigma$, and 
$$\mathcal{G}_1=\{\textit{G}=(\varSigma_P \cup \varSigma_T,E,\textit{Wt})~|~ \sigma = |\varSigma_P|=|\varSigma_T|, \textit{Wt}(G)= \sigma \}.$$
%Then the weight of graphs, $\textit{Wt}(E)=m=r$ where  $0 < r \leq \sigma$.
Since for any graph $\textit{G}=(\varSigma_P \cup \varSigma_T,E,\textit{Wt}) \in \mathcal{G}_1$, $|\varSigma_P|=|\varSigma_T|=\sigma$ and $ \textit{Wt}(G)= \sigma$, therefore $\min_{G \in \mathcal{G}_1}\{\textit{Wt}(\textit{mwm}(G)) \} = 1=q+1$.
%
%Let $q=1$. Then $m=q+r \geq \sigma$ because $0 < r \leq \sigma$ and $m \geq \sigma$, and 
%$$\mathcal{G}=\{\textit{G}(\varSigma_P \cup \varSigma_T,E,\textit{Wt})~|~ \sigma = |\varSigma_P|=|\varSigma_T|, \textit{Wt}(E)= \sigma \}.$$
%
%
% 
\item[Induction Hypothesis:] 
Assume that for $q=i$, 
$\min_{G \in \mathcal{G}_{i+1}}\{ \textit{Wt}(\textit{mwm}(G)) \} = i+1$, where
% \geq \sigma$ because $0 < r \leq \sigma$ and $m \geq \sigma$,
\begin{align*}
& m = i\sigma+r, ~\text{and}\\
& \mathcal{G}_{i+1} =\{\textit{G}=(\varSigma_P \cup \varSigma_T,E,\textit{Wt})~|~ \sigma = |\varSigma_P|=|\varSigma_T|, \textit{Wt}(G)=i\sigma+r \}.
\end{align*}
Let $\mathcal{G}'_{i+1}=\{G \in \mathcal{G}_{i+1}~|~\textit{Wt}(\textit{mwm}(G)) = i+1\}$. %G_{min} Note that, 
The set $\mathcal{G}'_{i+1}$ is non-empty by the Theorem~\ref{capsm:Th:min_MWBMs_Existence}.
%The non-empty existence such graphs is shown in the Theorem~\ref{capsm:Th:min_MWBMs_Existence}
We use this set in the following inductive step.
\item[Inductive Step:]
Let  $q=i+1$. We have to prove that
$\min_{G \in \mathcal{G}_{i+2}}\{ \textit{Wt}(\textit{mwm}(G)) \} = i+2$, where
\begin{align*}
 &m = (i+1)\sigma+r,~\text{and}\\
 &\mathcal{G}_{i+2}=\{\textit{G}=(\varSigma_P \cup \varSigma_T,E,\textit{Wt})~|~ \sigma = |\varSigma_P|=|\varSigma_T|,~\textit{Wt}(G)=(i+1)\sigma+r \}.
\end{align*}
%\begin{align*}
%m &= (i+1)\sigma+r
%% \geq \sigma$ because $0 < r \leq \sigma$ and $m \geq \sigma$, 
%\quad\text{and}\\
%\mathcal{G}_{i+2} &=\{\textit{G}=(\varSigma_P \cup \varSigma_T,E,\textit{Wt})~|~ \sigma = |\varSigma_P|=|\varSigma_T|,\\& \hspace*{5.1cm}\textit{Wt}(G)=(i+1)\sigma+r \}.
%\end{align*}
%
The existence of a graph $G \in \mathcal{G}_{i+2}$ with $\textit{Wt}(\textit{mwm}(G)) = i+2$ is proved in Theorem~\ref{capsm:Th:min_MWBMs_Existence}. 
Therefore, we only have to prove that there does not exist any graph in $\mathcal{G}_{i+2}$ whose weight of a maximum weight matching is $i+1$. 
%$\textit{Wt}(\textit{mwm}(G_1)) = i+1$.
Let us prove it by contradiction. Suppose there exists a graph $ G_* \in \mathcal{G}_{i+2}$ such that $\textit{Wt}(\textit{mwm}(G_*)) = i+1$.

%Note that, the set $\mathcal{G}'_{i+1}$ is non-empty by the 

Observe that, for any graph in $\mathcal{G}_{i+2}$, 
its weight is equal to $m=(i+1)\sigma+r=(i\sigma+r) + \sigma$.
%$\textit{Wt}(E)=(i+1)\sigma+r=(i\sigma+r) + \sigma$. 
Therefore, any graph in  $\mathcal{G}_{i+2}$ is generated by adding a total of $\sigma$ weight to the non-negative weight edges of a graph in $\mathcal{G}_{i+1}$. 
%Or in other way, $\mathcal{G}_{i+2}$ can be formed by adding $\sigma$ unit weight edges on the graphs of $\mathcal{G}_{i+1}$.(See Observation~\ref{capsm:Obs2}).

Therefore, $G_*$ can only be constructed from a graph in $\mathcal{G}'_{i+1}$ by adding a total of $\sigma$ weight to the non-negative weight edges of that graph in $\mathcal{G}'_{i+1}$;  
because for all $G \in \mathcal{G}_{i+1}\setminus\mathcal{G}'_{i+1}$, $\textit{Wt}(\textit{mwm}(G)) > i+1$.
%, because each of the graph in $\mathcal{G}_{i+1}\setminus\mathcal{G}'_{i+1}$ has maximum weight matching of weight at least $i+2$. 
%
Let $\varSigma = \{e_1,e_2,e_3, \ldots\}$ be the edges, where $\sigma=\sum_{e_i \in \varSigma}{\textit{Wt}(e_i)}$, whose weights are increased in $G \in \mathcal{G}'_{i+1}$ to build  $G_*$.
%
%Let $G \in \mathcal{G}'_{i+1}$ be any graph and $M=\textit{mwm}(G)$. So $\textit{Wt}(M)=i+1$.
%Let $G_*$ is constructed by adding $\sigma$ more weight to a graph in $\mathcal{G}'_{i+1}$.
 
%Hence $G_*$ can only be constructed by adding $\sigma$ unit weight edges to a graph in $\mathcal{G}'_{i+1}$ ; because $\forall G \in \mathcal{G}_{i+1}\setminus\mathcal{G}'_{i+1}$, $\textit{Wt}(E) > i+1$. 
%We use $\varSigma$ to denote the edges all the $\sigma$ unit weight edges, which are added  in $G \in \mathcal{G}'_{i+1}$ to build $G_*$. 
%
%
%
%\begin{itemize}
%\item
%
%{\bf Case 1:} 
\paragraph{\bf\textit{Case 1:}}
Let  $G \in \mathcal{G}'_{i+1}$ and $M = \textit{mwm}(G)$. 
If there exists at least one edge $e$ in $\varSigma$ such that
$e \in M$  
or if both the end points of $e$ are unmatched vertices, 
then let $M'=M \cup \{e\}$, which is a weighted matching of $G_*$, not necessarily of maximum weight. 
Therefore $$\textit{Wt}(\textit{mwm}(G_*)) \geq \textit{Wt}(M') = \textit{Wt}(M)+ \textit{Wt}(e) = i+1+\textit{Wt}(e) > i+1$$ which is a contradiction because we assumed that $(\textit{mwm}(G_*))=i+1$.

\textit{Note}: Hence for the rest of the cases we assume that 
none of the edges in $\varSigma$, which are added in $G \in \mathcal{G}'_{i+1}$ to get the $G_* \in \mathcal{G}_{i+2}$, belongs to $M$; or both the end points of none of the edges in $\varSigma$ are unmatched vertices. 
Therefore if $e=\{u,v\} \in \varSigma$, then: 
(a) $u$ is an unmatched vertex and $v$ is a matched vertex or vice versa, or
(b) both $u$ and $v$ are matched vertices, but $e \notin  M = \textit{mwm}(G)$.
%
%\item {\bf Case 2:} 
\paragraph{\bf\textit{Case 2:}}
Let there exists at least one edge $e=\{u,v\} \in \varSigma$ such that $\textit{Wt}(e)=w_\sigma \geq 2$.
Then we have the following two sub-cases which  are shown in Figure~\ref{capsm:Fig3}. Let $G \in \mathcal{G}'_{i+1}$ and $M = \textit{mwm}(G)$.

%---------------------------------
\begin{figure*}[!]
\centering
\includegraphics[angle=0,width=.9972\textwidth, trim=0 -5 0 5]{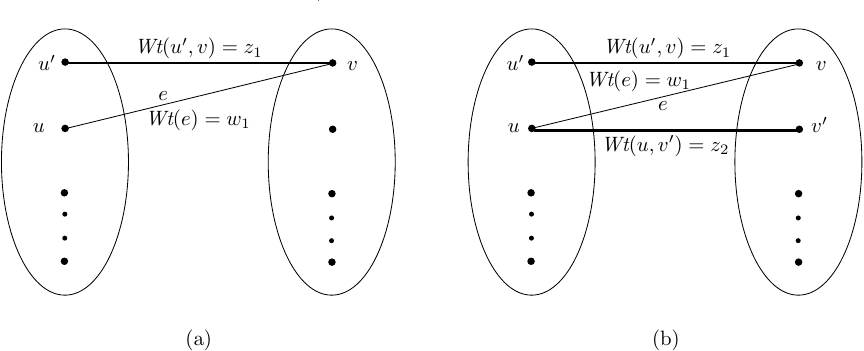}
\caption[Pictorial representation of the Sub-case~2(a) and Sub-case~2(b) in Theorem~\ref{capsm:Th:min_MWBMs}.]
%None of the $\sigma$ edges edges belong to $M$. Let the $\sigma$ edges form a matching in $G_*$. Then we have the following $2$ sub-cases. 
{{\bf(a)} This graph gives a pictorial representation of the Sub-case~2(a) in Theorem~\ref{capsm:Th:min_MWBMs}. {\bf(b)} Sketch of the graph considered in Sub-case~2(b) is shown here. In both the graphs the thick edges are maximum weight matching edges.}
\label{capsm:Fig3}
\end{figure*}
%---------------------------------

\textbf{Sub-case 2(a):} Assume that $u$ and $v$ be the unmatched and matched vertices in $G \in \mathcal{G}'_{i+1}$, respectively. 
So there exists an  edge $e'=\{u',v\} \in M$ which is incident on the matched vertex $v$.
Let $\textit{Wt}(e')=\textit{Wt}(u',v)=z_1$ and $\textit{Wt}(e)=\textit{Wt}(u,v)=w_1$ in the $G$. Therefore $z_1 \geq w_1$. 
Now add  the edge $e$ (or increase the edge weight of $e$) in $G$ where $\textit{Wt}(e)=w_\sigma \geq 2$ in order to generate 
$ G_* \in \mathcal{G}_{i+2}$ such that $\textit{Wt}(\textit{mwm}(G_*)) = i+1$.

If $ z_1 < w_1 +w_\sigma$, then let
%$$M \setminus \{e'\} \cup \{e\} \subseteq mwm(G_*).$$ 
$$M' = M \setminus \{e'\} \cup \{e\}$$ which is a weighted matching of $G_*$.
%$$mwm(G_*)=M \setminus \{(u',v)\} \cup \{(u,v)\}.$$ 
Hence $$\textit{Wt}(\textit{mwm}(G_*)) \geq \textit{Wt}(M')
= \textit{Wt}(M)-z_1+w_1+w_\sigma
= (i+1)-z_1+w_1+w_\sigma
>i+1$$ which is a contradiction.

Or else,  $$z_1 \geq w_1+w_\sigma ~\Leftrightarrow~ z_1-1 \geq w_1+(w_\sigma-1).$$
Therefore 
%$G \setminus (z_1-1) \cup (w_1+1)$ 
we can construct a new graph $G'$ from $G$ by decreasing one unit weight of the edge $e'=\{u',v\} \in M$ and increasing the weight of the edge $e=\{u,v\} \notin M$  by  one unit in $G$.
As a consequence, the weight of $G'$ remains the same as that of $G \in \mathcal{G}'_{i+1}$ and so $G' \in \mathcal{G}_{i+1}$. %\mathcal{G}'_{i+1}$.
%and so $G' \in \mathcal{G}'_{i+1}$. 
But  
$$\textit{Wt}(\textit{mwm}(G')) = i < \textit{Wt}(M) = i+1$$ 
which contradicts  the induction
hypothesis that $\min_{G \in \mathcal{G}_{i+1}}\{ \textit{Wt}(\textit{mwm}(G)) \} = i+1$.
%definition of the set $\mathcal{G}'_{i+1}$.

\textbf{Sub-case 2(b):} Suppose both $u$ and $v$ are matched vertices but $e=\{u,v\} \notin M$. See Figure~\ref{capsm:Fig3}(b).
So there exist two edges $e'=\{u',v\} \in M$ and $ e''=\{u,v'\} \in M$ which are incident on the matched vertices $v$ and $u$, respectively.
Let $\textit{Wt}(e')=z_1$, $\textit{Wt}(e'')=z_2$ and $\textit{Wt}(e)=w_1$ in $G \in \mathcal{G}'_{i+1}$. 
%Therefore 
$$So,~ z_1+z_2 \geq w_1\quad\text{in } G.$$
Now after adding the edge $e$ in $G$ with $\textit{Wt}(e)=w_\sigma \geq 2$, if 
$$ z_1+z_2 < w_1 +w_\sigma,$$ then let
%$$M \setminus \{(e',e''\} \cup \{e\} \subseteq mwm(G_*).$$ 
%%$$mwm(G_*)=M \setminus \{(u',v)\} \cup \{(u,v)\}.$$ 
$$M' =M \setminus \{e',e''\} \cup \{e\}$$ which is a weighted matching of $G_*$.
Hence $$\textit{Wt}(\textit{mwm}(G_*)) \geq \textit{Wt}(M') 
= \textit{Wt}(M)-z_1-z_2+w_1+w_\sigma 
= (i+1)-z_1-z_2+w_1+w_\sigma 
> i+1 $$ which is a contradiction.

Or else,  $$z_1+z_2 \geq w_1+w_\sigma ~ \Leftrightarrow ~ (z_1-1)+z_2 \geq w_1+(w_\sigma-1).$$
%
%$$\text{ if } z_1 +z_2 < w_1 +2,\text{ then } mwm(G_*)=M \setminus \{(u',v),(u,v')\} \cup \{(u,v)\}.$$ 
%Hence $\textit{Wt}(\textit{mwm}(G_*)) > \textit{Wt}(M)$ which is a contradiction.
%$$\text{Else, we can write } z_1+z_2 \geq w_1+2 \Leftrightarrow (z_1-1) +z_2 \geq w_1+1.$$
Therefore 
%$G \setminus (z_1-1) \cup (w_1+1)$ 
we can construct a new graph $G'$ from $G$ by reducing one unit weight of the edge $e'=\{u',v\} \in M$ and adding one unit weight to the edge $e=\{u,v\} \notin M$ %of unit weight
 of $G$.
As a consequence, the weight of $G'$ is the same as that of $G \in \mathcal{G}'_{i+1}$ and so $G' \in \mathcal{G}_{i+1}$. 
But  
$$\textit{Wt}(\textit{mwm}(G'))=i < \textit{Wt}(M) = i+1$$ 
which contradicts the 
hypothesis that $\min_{G \in \mathcal{G}_{i+1}}\{ \textit{Wt}(\textit{mwm}(G)) \} = i+1$.
%definition of the set $\mathcal{G}'_{i+1}$.
%
%
\paragraph{\bf\textit{Case 3:}} 
Let for each edge $e\in \varSigma$, $\textit{Wt}(e)= 1$. 
Consider $\varSigma=\{e_1=\{u_1,v_1\},e_2=\{u_2,v_2\}, \ldots, e_\sigma=\{u_\sigma,v_\sigma\}\}$ and their respective weights in $G \in \mathcal{G}'_{i+1}$ are given by $\{w_1,w_2,\ldots,w_\sigma\}$.
We add these $\sigma$ number of edges of $\varSigma$ in $G \in \mathcal{G}'_{i+1}$ to produce a graph
$ G_* \in \mathcal{G}_{i+2}$ such that $\textit{Wt}(\textit{mwm}(G_*)) = i+1$.
Further let $M = \textit{mwm}(G)$.
%

%\textit{Claim}: 
Therefore, there must exist two edges in $\varSigma$ which are not adjacent.
Because if not, then  all the edges of $\varSigma$ are adjacent to one vertex. Without loss of generality, suppose $u_1=u_2= \cdots = u_\sigma$. See Figure~\ref{capsm:Fig5} and consider the following two possibilities.
\begin{enumerate}[(a)]
\item
If $u_1 \in \varSigma_P$ is an unmatched vertex in $G \in \mathcal{G}'_{i+1}$, then there must be another unmatched vertex in $\varSigma_T$ of the graph $G$, because $\sigma = |\varSigma_P|=|\varSigma_T|$. Say the unmatched vertex is $v_1 \in \varSigma_T$. 
If we add an edge $\{u_1,v_1\} \in \varSigma$ in $G \in \mathcal{G}'_{i+1}$, then this kind of graph is already addressed in Case~1. 
Therefore,  at most $\sigma-1$ number of edges of unit weight can be added in $G$ while generating the $G_*$. 
%But in $\varSigma$ we have  $\sigma$ number of unit weight distinct edges which are to be added in $G$ to built the $G_*$. 
This is a contradiction. 
\item
Similarly, if $u_1 \in \varSigma_P$ is a matched vertex in $G \in \mathcal{G}'_{i+1}$, then there must be another matched vertex in $\varSigma_T$ of the graph $G$. The rest of the argument is similar to the previous unmatched case.
\end{enumerate}

%---------------------------------
\begin{figure*}[htt]
\centering
%trim={<left> <lower> <right> <upper>}
\includegraphics[angle=0,width=.632\textwidth, trim={0 -5 0 0}]{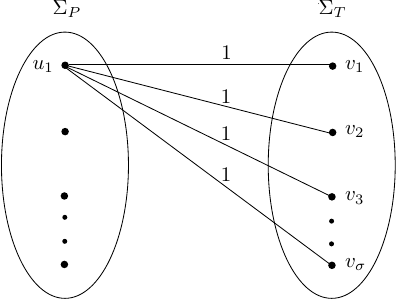}
\caption[Kind of graph does not arise in  Case~3 of Theorem~\ref{capsm:Th:min_MWBMs}.]
{There must exists two edges $e_1,e_2 \in \varSigma$ such that $e_1$ and $e_2$ are not adjacent. This kind of graph does not arise in Case~3 of Theorem~\ref{capsm:Th:min_MWBMs}.}
\label{capsm:Fig5}
\end{figure*}
%---------------------------------

So we assume the two non-adjacent edges be $e_1, e_2 \in \varSigma$.
% So we assume that the two edges be $e_1,e_2 \in \varSigma$  such that $e_1$ and $e_2$ are not adjacent.
Then a maximum of four edges in $M$ are adjacent to edges $e_1,e_2 \in \varSigma$. Let $e_1',e_2',e_3',e_4'$ be such edges and  $z_1,z_2,z_3,z_4$ be their corresponding weights in $G \in \mathcal{G}'_{i+1}$, respectively.  
Therefore, $$ ~ z_1 +z_2+z_3+z_4 \geq w_1 +w_2 \quad\text{in }  G.$$
Now after adding $\sigma$ edges of $\varSigma$ in $G$, if
$$z_1 +z_2+z_3+z_4 < w_1 +w_2 + 2,$$
then let
%$$\textit{mwm}(G_*)=M \setminus \{e_1',e_2',e_3',e_4'\} \cup \{e_1,e_2\}.$$
$$M'=M \setminus \{e_1',e_2',e_3',e_4'\} \cup \{e_1,e_2\}$$
which is a weighted matching of $G_*$.
Hence 
\begin{align*}
\textit{Wt}(\textit{mwm}(G_*)) 
& \geq \textit{Wt}(M')\\
& = \textit{Wt}(M)-(z_1+z_2+z_3+z_4)+(w_1+w_2+2)\\
& > i+1
\end{align*}
which is a contradiction.
 
Or else,  
\begin{align*}
 & z_1+z_2+z_3+z_4 \geq w_1+w_2+2 \\ \Leftrightarrow 
~& z_1 + z_2+z_3+z_4-1 \geq w_1+w_2+1.
\end{align*}
As a consequence, by similar argument as stated in Sub-case~2(b), we can construct a new graph $G'$ %from $G$, 
whose weight is same as that of $G \in \mathcal{G}'_{i+1}$ and  so $G' \in \mathcal{G}_{i+1}$. 
But  
$$\textit{Wt}(\textit{mwm}(G'))=i < \textit{Wt}(M) = i+1$$ 
which contradicts the induction
hypothesis that $\min_{G \in \mathcal{G}_{i+1}}\{ \textit{Wt}(\textit{mwm}(G)) \} = i+1$.
%which contradicts the definition of the set $\mathcal{G}'_{i+1}$.
%\end{itemize}
\end{description}
%Hence the proof.% by the principle of mathematical induction.
This completes the proof.
\end{proof}

%---------------------------
An equivalent statement of the  Theorem~\ref{capsm:Th:min_MWBMs} is the following.
% Here the tight lower bound is written in terms of the function of the internal veriable $m$ and $\sigma$.
\begin{corollary}
\label{capsm:Cor:min_MWBMs}
For the partition  %$\mathcal{G}= 
$ \Gg \equiv \{\textit{G}=(\varSigma_P \cup \varSigma_T,E,\textit{Wt})~|~ \sigma = |\varSigma_P|=|\varSigma_T|, m=\textit{Wt}(G) \text{ and } m \geq \sigma \}$ 
%then
$$ \min\limits_{G \in \Gg}\{ \textit{Wt}(\textit{mwm}(G)) \} = \Big\lceil  \frac{m-\sigma}{\sigma}  \Big\rceil+1. $$
\end{corollary}
\begin{proof}
Since $m  \geq \sigma$, we can always write $m$ as 
%a linear combination of  $\sigma$ and~$1$. Observe that, if %we write 
$q \sigma + r$ for some $q,\,r \in \nint$ where $0 < r \leq \sigma$. Then the term $\lceil  \frac{m-\sigma}{\sigma}  \rceil$ can be written as 
\[\Big\lceil  \frac{m-\sigma}{\sigma}  \Big\rceil 
=\Big\lceil  \frac{q \sigma + r-\sigma}{\sigma}  \Big\rceil
=\Big\lceil  \frac{(q-1) \sigma + r}{\sigma}  \Big\rceil
= (q-1)+1=q.
\]
%Therefore, proof of this theorem is equivalent to the proof of the second next Theorem~\ref{capsm:Th:min_MWBMs_2} stated below. 
Hence the statement in this corollary is equivalent to 
%the same in 
 Theorem~\ref{capsm:Th:min_MWBMs}.
\end{proof}

The following theorem is for the partition of graphs in $\Gl$. 
The proof is trivial. Note that for  $0<m<\sigma$, the term 
%$\Big\lceil  \frac{m-\sigma}{\sigma}  \Big\rceil=0$.
$\Big\lceil  \frac{m-\sigma}{\sigma}  \Big\rceil+1=1$.

\begin{theorem}
For the partition  %$\mathcal{G}= 
$ \Gl \equiv \{\textit{G}=(\varSigma_P \cup \varSigma_T,E,\textit{Wt})~|~ \sigma = |\varSigma_P|=|\varSigma_T|, m=\textit{Wt}(G) \text{ and } m < \sigma \}$ 
%then
$$ \min\limits_{G \in \Gl}\{ \textit{Wt}(\textit{mwm}(G)) \} = 1.$$
%\Big\lceil  \frac{m-\sigma}{\sigma}  \Big\rceil+1. $$
\end{theorem}

%================================================================================================
\section{Conclusion}
\label{Sec:tlb_conclusion}
Thus we have given a tight lower bound  $\lceil  \frac{m-\sigma}{\sigma} \rceil+1$ for the weights of maximum weight matching of bipartite graphs each having fixed weight as $m$ and size of a vertex partition as $\sigma$. 
The above finding may be applied in stringology,  monitoring computer network, computer vision, pattern recognition, machine learning, compiler design with cloud architecture and telecommunication network to minimize the cost, time with least traffic load and critical path routing. 

%%================================================================================================
%\begin{acknowledgements}
%I would like to thank the Associate Editor for his helpful advice, which helped me to improve the manuscript.
%\end{acknowledgements}
%%================================================================================================

%\begin{acknowledgements}
%\end{acknowledgements}

%\bibliographystyle{LaTeX/spbasic}      % basic style, author-year citations
%\bibliographystyle{LaTeX/spmpsci}      % mathematics and physical sciences
%\bibliographystyle{LaTeX/spphys}       % APS-like style for physics
\bibliographystyle{plain}   
%\bibliographystyle{apalike} 
%\bibliography{../../../thesis_new/references}
\bibliography{references}

\end{document}